\documentclass[pra,showpacs,
twocolumn
]{revtex4}
\usepackage{amsmath,amsthm,amssymb,amsfonts,graphicx}

\begin{document}

\title{Quantum anonymous voting with anonymity check}
\author{Dmitri Horoshko}
\email[E-mail: ]{dhoroshko@rambler.ru} \affiliation{B.I.Stepanov
Institute of Physics, NASB, Nezavisimosti Ave. 68, Minsk 220072
Belarus}
\author{Sergei Kilin}
\affiliation{B.I.Stepanov Institute of Physics, NASB,
Nezavisimosti Ave. 68, Minsk 220072 Belarus}
\date{\today}
\begin{abstract}
We propose a new protocol for quantum anonymous voting having
serious advantages over the existing protocols: it protects both
the voters from a curious tallyman and all the participants from a
dishonest voter in unconditional way. The central idea of the
protocol is that the ballots are given back to the voters after
the voting process, which gives a possibility for two voters to
check the anonymity of the vote counting process by preparing a
special entangled state of two ballots. Any attempt of cheating
from the side of the tallyman results in destroying the
entanglement, which can be detected by the voters.
\end{abstract}
\pacs{03.67.-a, 03.67.Dd} \maketitle

There is a general trend in the modern society to automatization
and computerization of nearly all aspects of social life,
including such subtle area as voting procedure in various
contexts: from state governmental elections to decision making in
rather small groups like parliaments or councils. As a
consequence, a number of protocols for electronic voting have been
developed and successfully applied in the last decades
\cite{Gritzalis}. Since such protocols meet the information
security problems of confidentiality, authentication and data
integrity, they belong to the scope of the science of
cryptography. In the modern electronic voting systems the
information security is provided by means of public-key
cryptography, guaranteeing secrecy under condition of limited
computational resources of a potential adversary. With the advent
of quantum computers \cite{KilinRev} this condition becomes
impractical, thus inspiring interest in unconditionally secure
voting schemes and protocols. One perspective way to this end is
connected to using quantum systems as information carriers, which
proved to be successful for the development of unconditionally
secure key distribution; the technology known as quantum key
distribution \cite{GisinRev} has reached presently the level of
commercial realizations.

In the present work we propose a protocol of anonymous
binary-valued voting involving $n$ persons (voters), each making a
binary decision $b_i\in\{0,1\}$ and writing it on a ballot, and
one person (tallyman) collecting the ballots and announcing the
result $s=\sum_i b_i$. The proposed protocol possesses two
security properties. The first property is the ''anonymity of
voting'', meaning that the value of individual vote of the $i$th
voter, $b_i$, remains unknown to other voters, the tallyman, and
any third party possibly monitoring the communication lines,
unless $s=0$, $s=n$, $min\{s,n-s\}$ voters cooperate, or the $i$th
voter discloses his decision. The second property may be called
''non-exaggeration'' and means inability of a voter to contribute
a number different from $0$ or $1$ to the final sum $s$. The
anonymity of voting protects the voters from a curious tallyman
(and other parties), who may wish to learn who voted in which way,
while the non-exaggeration protects the entire community from
malicious voters who may wish to vote twice. The proposed protocol
includes operations with quantum systems and provides
unconditional security for both anonymity and non-exaggeration,
which distinguishes it from other existing voting protocols, both
quantum and classical, briefly reviewed below.

There is a possibility of guaranteeing the anonymity of voting
unconditionally by means of conventional, i.e. classical
cryptography, based on mathematical encryption. The corresponding
voting protocol \cite{ChaumEuroCrypt88} is based on the principle
of ''sender untraceability'', meaning such a communication scheme,
where the recipient of several messages from several senders
cannot determine which message came from which sender. Such a
communication can be realized with unconditional security in the
sense that the recipient is unable to establish any relation
between the messages and the senders, even being in possession of
infinite computational power \cite{ChaumJCrypt88}. However, the
very property of untraceability creates, in the case of voting, an
additional problem of determining which ballots come from legal
voters, since illegal participants can send ballots in an
untraceable way. This problem is solved by a special ''ballot
issuing'' protocol (based on the technique of ''blind signature'')
providing each legal voter with an ''unforgeable'' and ''blind''
digital ballot, which is used for sending a vote. The term
''unforgeable'' means that the ballot cannot be cloned, while the
term ''blind'' means that the ballots are in no way related to the
identities of legal voters. The ballots in the ballot issuing
protocol are unconditionally ''blind'' but only conditionally
''unforgeable'', that is a person in possession of rich enough
computational power is able to vote instead of legal voters. Thus,
the property of ''non-exaggeration'' is realized by the overall
voting protocol in a conditional way only.

A quantum protocol for anonymous surveying has been proposed
recently, whose aim is to calculate the sum of individual
contributions of the participants, like in the voting protocol,
but with the contributions being real numbers from some limited
interval rather than binary digits \cite{Vaccaro05}. The protocol
is based on bipartite entangled quantum state, whose relative
phase carries the sum of contributions and can be measured only
when two parts of the entangled system are gathered in the same
location. Application of this protocol to binary-valued voting
meets the problem of multiple voting by a dishonest voter, which
is proposed to be solved by employing two non-cooperative ballot
agents (tallymen). A similar protocol for quantum voting has been
proposed \cite{Buzek05}, meeting the same problem of multiple
voting, and various ways for solving it have been discussed
requiring also an employment of two non-cooperative tallymen.
Another similar quantum protocol based on multipartite
entanglement and quantum Fourier transform has been proposed
\cite{Dolev06}, which is also vulnerable to multiple voting and
may be securely applied only under the assumption that the voter
has no full control of the ballot at the time of writing the
choice. All the mentioned quantum protocols provide unconditional
anonymity of voting, but the property of ''non-exaggeration'' is
reached on the cost of serious additional assumptions, which may
be viewed impractical in some applications. In contrast, the
present protocol provides in unconditional way both
''eavesdropping detection'', meaning non-zero probability of
detection of any attempt to learn the distribution of votes among
the voters, and ''non-exaggeration'', thus protecting the protocol
from dishonest voters from one side and dishonest tallyman from
the other side. Here we consider, like in Refs.
\cite{ChaumEuroCrypt88,Vaccaro05,Buzek05,Dolev06}, a curious but
not malicious tallyman, whose dishonest action is limited to
learning the distribution of votes among the voters, but not to
announcing a wrong value of the voting result $s$.

The protocol of voting is as follows. The participants are $n$
legal voters labelled by index $i=1,2,...,n$ and a tallyman.
\begin{enumerate}
\item Each voter chooses either to vote or to check the anonymity
of voting.
\begin{enumerate} \item In the case of voting the voter makes a binary decision $b_i$ with
$b_i=0$ corresponding to ''no'' and $b_i=1$ corresponding to
''yes'' decision, and encodes it into a state of a two-level
quantum system -- qubit -- playing the role of a ballot. Two
orthogonal states $|0\rangle_i$ and $|1\rangle_i$ of a qubit
(computational basis) are used for encoding of the corresponding
value of $b_i$. \item In the case of anonymity check the $i$th
voter cooperates with the $j$th voter, who also chooses to check
the anonymity, and they together prepare their pair of qubits in
the Bell state $|\Psi^{+}\rangle_{ij}$, where the Bell states are
defined as
\begin{eqnarray}
|\Phi^{\pm}\rangle_{ij}&=&\frac1{\sqrt{2}}\{|0\rangle_i|0\rangle_j\pm
|1\rangle_i|1\rangle_j\},\\
|\Psi^{\pm}\rangle_{ij}&=&\frac1{\sqrt{2}}\{|0\rangle_i|1\rangle_j\pm
|1\rangle_i|0\rangle_j\}.
\end{eqnarray}
\end{enumerate}
\item After the encoding all voters send their qubits to the
tallyman together with their identities. The latter excludes the
possibility of voting for illegal participants and the possibility
for legal voters to vote instead of their colleagues. \item The
tallyman collects all $n$ qubits and calculates the number of
''yes'' votes by applying to the $n$-qubit system the projector
valued measure (PVM)
\begin{equation} \label{PVM}
\hat P(s)=\sum_{\pi} |m(s,\pi)\rangle\langle m(s,\pi)|,
\end{equation}
where $|m(s,\pi)\rangle$ is a product state of $n$ qubits in the
computational basis, having exactly $s$ 1's in the order
determined by the permutation variable $\pi$. The tallyman
announces the voting result ''s votes yes''. \item The tallyman
sends the qubits back to the voters. \item The voters make a
ballot test.
\begin{enumerate}
\item The voters, who have chosen to vote, measure their qubits in
the computational basis. If the state of the qubit is different
from the sent one, they state the ballot test failure. \item The
voters, who have chosen to make an anonymity check, make a
measurement of their pair of qubits in the Bell basis. If they get
a result which is different from the Bell state
$|\Psi^{+}\rangle_{ij}$, they state the ballot test failure.
\end{enumerate}
\end{enumerate}

A few comments to the protocol are necessary. In the present
protocol the statement of the ballot test failure does not mean a
public accusation of the tallyman, it is rather an information for
the personal use by the voter (e.g. a council member).

The numbering with $\pi$ is as follows. All $n$-bit strings with
exactly $s$ 1's represent numbers $0\le m\le (2^n-1)$ in binary
notation. Let us sort the strings in increasing order of the
corresponding numbers $m$ and label them with index $\pi$ taking
consecutive integer values from 1 to $d_s={n\choose s}$. In this
way for any $0\le s\le n$ we get a set of strings $m(s,\pi)$. The
product state of $n$ qubits in computational basis with individual
qubit states $|b_i\rangle_i$, $b_i$ being $i$th bit from the
string $m(s,\pi)$, is the state $|m(s,\pi)\rangle$. For example,
in the case of 5 qubits $m(1,2)=00010$ and
$|m(1,2)\rangle=|0\rangle_5|0\rangle_4|0\rangle_3|1\rangle_2|0\rangle_1$.
The states $|m(s,\pi)\rangle$ are mutually orthogonal.

The projector given by Eq.(\ref{PVM}) is a projector on the
subspace of $n$-qubit system, having $s$ states $|1\rangle$ and
$n-s$ states $|0\rangle$. Let us denote this $d_s$-dimensional
subspace $V_s$. It is easy to see, that the subspaces
corresponding to different values of $s$ are orthogonal and their
sum is the entire state space of $n$ qubits. The states
$|m(s,\pi)\rangle$ for given $s$ form a basis in $V_s$. The
application of projective measurement Eq.(\ref{PVM}) corresponds
to measuring the number of ''yes'' votes, but not their
distribution among the voters.

Let us see how the protocol guarantees the anonymity of voting.
Consider an event $E(\mu)$ consisting in $2k$ voters choosing to
check the anonymity, $l$ voters voting ''yes'' and the rest voting
''no''. The state of $n$ ballot qubits collected by the tallyman
is represented by a state
\begin{equation}\label{E}
|E(\mu)\rangle=\frac1{\sqrt{2^k}}\sum_{\pi\in\Omega
(\mu)}|m(k+l,\pi)\rangle,
\end{equation}
where $\Omega (\mu)$ is a set of $2^k$ possible values of $\pi$.
The state Eq.(\ref{E}) belongs to the subspace $V_{k+l}$ and
therefore is not affected by the projective measurement defined by
Eq.(\ref{PVM}). In the absence of errors the qubits sent back to
the voters will always pass the ballot test in the Step 5.

Now we consider a curious tallyman, who makes an additional
measurement of qubits with the aim to obtain some information on
who voted which way. The simplest way to learn the vote of the
$i$th voter is just to measure the $i$th qubit in computational
basis. If the $i$th voter has chosen to vote, this attack passes
unnoticed. But, if the $i$th voter has chosen to check the
anonymity with the $j$th voter, their state
$|\Psi^{+}\rangle_{ij}$ will be transformed into
$|0\rangle_i|1\rangle_j$ or $|1\rangle_i|0\rangle_j$ with equal
probabilities, and the subsequent Bell measurement will give
results $|\Psi^{+}\rangle_{ij}$ or $|\Psi^{-}\rangle_{ij}$ with
probabilities $\frac12$. The latter result means the anonymity
check failure. Thus, a curious tallyman faces a risk of being
detected.

The possible attacks from a curious tallyman are in no way
restricted to measurement of single qubits. The tallyman may wish
to learn some partial information concerning the distribution of
votes, for example, the total number of ''yes'' votes from a
fraction of the voters. As it was mentioned above, we consider a
curious but not malicious tallyman, who follows the protocol up to
projecting the qubits onto a subspace $V_s$ and correctly
determining the value of $s$. After that the tallyman may be
interested in making an additional measurement of the qubits. To
prove the unconditional ''eavesdropping detection'' we need to
show that for any such measurement there is an event $E(\mu)$ for
which the probability of ballot test failure is non-zero.

The most general type of measurement on a system of n qubits,
which we call ''the object'', consists in attaching to them
another quantum system of at least the same dimensionality (the
measuring apparatus), making a unitary transformation $U_{OA}$ of
both the object and the apparatus, and analyzing the resulting
state of the apparatus \cite{Wigner63}. Since the states
$|m(s,\pi)\rangle$ for given $s$ form a basis in $V_s$, the
unitary transformation can be determined by its action on the
basis states:
\begin{equation}\label{Wig}
U_{OA}|m(s,\pi)\rangle_O|a_0\rangle_A=\sum_{\pi'}{|m(s,\pi')\rangle_O|a_{\pi\pi'}\rangle_A},
\end{equation}
where $|a_0\rangle_A$ is the initial state of the apparatus, and
$|a_{\pi\pi'}\rangle_A$ are its final states, generally not
normalized. The subscripts $O$ and $A$ refer to the object and the
apparatus respectively. Here we suggest that the measurement does
not take the state of the qubits outside the subspace $V_s$,
because otherwise the non-zero probability of ballot test failure
is obvious. Thus, all possible measurements of the tallyman are
parameterized by a set of states $\{|a_{\pi\pi'}\rangle_A\}$.

To prove the property of ''eavesdropping detection'' of the
proposed protocol, we need a result concerning the structure of
strings $m(s,\pi)$ for given $s$. In the following we imply that
$s$ is fixed and the positions of bits in a string are numbered
from right to left.

\newtheorem*{lemma}{Lemma}
\begin{lemma}
For given $s$ and any two numbers $1\le \pi,\pi'\le d_s$, the
string $m(s,\pi')$ can be obtained from the string $m(s,\pi)$ by a
finite number of pairwise permutations of 0s and 1s.
\end{lemma}
\begin{proof}
Let $w(s,\pi,\pi')$ be the set of positions of bits, which are
different in $m(s,\pi)$ and $m(s,\pi')$. This set is a sum of two
non-overlapping subsets: $w_0(s,\pi,\pi')$, containing the
positions of bits which are equal to $0$ in $m(s,\pi)$, and
$w_1(s,\pi,\pi')$, containing the positions of bits which are
equal to $1$ in $m(s,\pi)$. The lengths of the subsets
$w_0(s,\pi,\pi')$ and $w_1(s,\pi,\pi')$ coincide, because the
number of 1s in both strings is the same. Let us make a set of
pairs $w_{01}(s,\pi,\pi')$ of the elements of both subsets, taking
one position from $w_0(s,\pi,\pi')$ and one position from
$w_1(s,\pi,\pi')$ in increasing order. The string $m(s,\pi)$
subjected to permutation of bits at positions defined by the set
$w_{01}(s,\pi,\pi')$ gives the string $m(s,\pi')$.
\end{proof}

Now we can proceed to proving the property of ``eavesdropping
detection'' of the proposed protocol, which is based on the
following theorem.

\newtheorem*{theorem}{Theorem}
\begin{theorem}
For any measurement, defined by the apparatus states
$\{|a_{\pi\pi'}\rangle_A\}$, where is an event $E(\mu)$ for which
the probability of ballot test failure is non-zero, unless all the
states satisfy
\begin{equation}
|a_{\pi\pi'}\rangle_A=|a_{11}\rangle_A\delta_{\pi\pi'},
\end{equation}
i.e. no measurement is done.
\end{theorem}
\begin{proof}
Let us suggest that the apparatus states contain a non-zero
off-diagonal state $|a_{\pi\pi'}\rangle_A$, $\pi\ne\pi'$. Consider
the event $E(\mu)$, where all voters have chosen to vote and the
distribution of votes corresponds to the string $m(s,\pi)$. For
this event the probability of ballot test failure is non-zero,
because the qubits received by the voters are in a mixture having
component $|m(s,\pi')\rangle$.

Now let us consider measurements with the apparatus states
satisfying
\begin{equation}
|a_{\pi\pi'}\rangle_A=|a_{\pi\pi}\rangle_A\delta_{\pi\pi'}.
\end{equation}
Consider any two values $\pi\ne\pi'$. Due to the Lemma the strings
$m(s,\pi')$ and $m(s,\pi)$ differ by finite number $k$ of pairwise
permutations determined by the set of bit position pairs
$w_{01}(s,\pi,\pi')$. Consider the event $E(\nu)$, where $k$ pairs
of voters, determined by $w_{01}(s,\pi,\pi')$, have chosen to
check the anonymity, and the rest have voted in a way described by
the coinciding bits of $m(s,\pi')$ and $m(s,\pi)$. For this event
the state of qubits before the measurement is a superposition of
$2^k$ states of the type of Eq.(\ref{E}), including
$|m(s,\pi)\rangle$ and $|m(s,\pi')\rangle$. After the interaction
with the apparatus these two components get factors
$|a_{\pi\pi}\rangle$ and $|a_{\pi'\pi'}\rangle$ respectively, as
indicated by Eq.(\ref{Wig}), which leads to a non-zero probability
of wrong result for Bell state measurement, unless
$|a_{\pi\pi}\rangle=|a_{\pi'\pi'}\rangle$.
\end{proof}

In summary, we have proposed a quantum protocol of voting,
guaranteeing that each voter contributes only one vote and that
any attempt of learning who voted which way is detectable with
non-zero probability. The protocol is a cryptographic primitive,
intended to be an element of a more complicated cryptographic
system, providing complex security of voting, including, e.g.
authentication of legal voters etc. The main weakness of the
protocol is its inability to realize a guaranteed anonymity of a
single voting act, providing only the probabilistic
''eavesdropping detection'', which is useful for application to
many voting acts during a rather long period. Another weakness is
connected to the necessity of cooperation of voters having
opposite decisions, i.e. most probably, belonging to different
fractions. However, to our knowledge, it is the first voting
protocol uniting the protection of the voters from a dishonest
tallyman and the protection of the participants from a dishonest
voter in unconditional way.

\begin{acknowledgments}
This work was supported by the project EQUIND performed within the
$6^{th}$ Framework programme of European Commission.
\end{acknowledgments}

\end{document}